\documentclass[12 pt]{article}
\usepackage[utf8]{inputenc}
\usepackage[english]{babel}
\usepackage{enumerate}
\usepackage[T1]{fontenc}
\usepackage{fullpage}

\usepackage[utf8]{inputenc}
\usepackage{tikz,inputenc,xcolor,amsthm,hyperref, cite, color, url,amsfonts,amsmath,amssymb,caption,multicol,graphicx,fancyhdr,tikz,pgfplots,comment,commath,bm}
\usepackage{pdflscape}
\usepgfplotslibrary{groupplots}

\DeclareMathOperator{\E}{\mathbb{E}}
\DeclareMathOperator{\Pb}{\mathbb{P}}

\newcommand{\vb}{\vspace{3mm}}
\newcommand{\f}[2]{\frac{#1}{#2}}

\theoremstyle{plain}

\newtheorem{theorem}{Theorem}
\newtheorem{proposition}{Proposition}

\theoremstyle{definition}

\theoremstyle{definition}

\theoremstyle{definition}
\newtheorem{algorithm}{Algorithm}
\theoremstyle{definition}
\newtheorem{example}{Example}
\theoremstyle{definition}

\usepackage{tikz}
\usetikzlibrary {positioning}
\usetikzlibrary{automata}
\definecolor {processblue}{cmyk}{1,1,1,1}

\tikzset{state/.style ={circle,draw,black,text=black,minimum width =1 cm}}
\tikzset{box/.style   ={rectangle, draw,black,text=black,minimum width =1 cm}}
\tikzstyle{phase} = [fill,shape=circle,minimum size=5pt,inner sep=0pt]

  \usepackage{pgfplots}
  \pgfplotsset{compat=newest}
  \usetikzlibrary{plotmarks}
  \usepackage{grffile}
  \usepackage{amsmath}

\usepackage[affil-it]{authblk}

\title{Predicting the confirmation time of Bitcoin transactions}
\author{D.T. Koops}
\affil{Korteweg-de Vries Institute, University of Amsterdam}
\date{\today}

\begin{document}
\maketitle 

\begin{abstract}
\noindent
We study the probabilistic distribution of the confirmation time of Bitcoin transactions, conditional on the current memory pool (i.e., the queue of transactions awaiting confirmation). The results of this paper are particularly interesting for users that want to make a Bitcoin transaction during `heavy-traffic situations', when the transaction demand exceeds the block capacity. In such situations, Bitcoin users tend to bid up the transaction fees, in order to gain priority over other users that pay a lower fee. We argue that the time until a Bitcoin transaction is confirmed can be modelled as a particular stochastic fluid queueing process (to be precise: a Cram\'er-Lundberg process). We approximate the queueing process in two different ways. The first approach leads to a lower bound on the confirmation probability, which becomes increasingly tight as traffic decreases. The second approach relies on a diffusion approximation with a continuity correction, which becomes increasingly accurate as traffic intensifies. The accuracy of the approximations under different traffic loads are evaluated in a simulation study.
\end{abstract}

\section{Introduction}
Bitcoin is a \textit{virtual currency}, in the sense that it can be used to transact or store value in a peer-to-peer network on the internet. The main innovation of Bitcoin is that it can operate without central authority, as opposed to traditional money accounts that are run by banks. The Bitcoin network has been running since 2009; which is when the `genesis block' was mined by an unknown person who refers to himself as Satoshi Nakamoto. Since then the popularity of Bitcoin has increased substantially. What became painfully obvious during December 2017 - January 2018, when the price of Bitcoin reached new all-time highs, is that the network is not capable of dealing with all transactions for low fees when the traffic intensity on the blockchain increases too much. During this period, Bitcoin users were paying as much as 50 US dollar for a simple transaction. The reason that users were paying such hefty fees, is that the network was clogged, and transactions with a higher fee get priority. It is clear that from a user perspective, it would be useful to have precise estimates of the confirmation time during such circumstances. This is the problem that is addressed in this paper.

\vb

Let us consider in some more detail how the Bitcoin blockchain works. Bitcoin relies on a protocol called \textit{proof-of-work}. This protocal is basically a voting mechanism that decides which transactions should be included in the {blockchain}, and it acts as a replacement of a trusted third party. A blockchain is a chain of blocks, each of which contains a bunch of transactions that are considered \textit{verified} or \textit{confirmed}. Newly found blocks can be linked to the last block in the blockchain (in principle blocks can also be linked to earlier blocks, but then they will not be accepted by other nodes as the longest chain is seen as valid). Computers with strong computational power are competing to solve a cryptographic puzzle as quickly as possible, in order to get the right to determine the next block. Computers that perform this task are called \textit{miners}, and they receive \textit{mining rewards} for being the first to mine a new block. The reference to mining comes from the fact that these miners have to search large outcome spaces, and they need to get lucky to find a block. In that sense it is somewhat comparable to e.g.\ gold mining, where a large area is searched and occasionally some gold is found. It should be noted that all this computational effort is spent purely as a voting procedure, in the sense that solving the puzzle allows one to submit a vote on what should be the next block (although the vote need not be accepted by the majority of other nodes, e.g.\ if it contains invalid transactions). Another, possibly more obvious, voting mechanism would be to randomly select a computer or IP address to propose the next block. The problem with this idea is twofold. Firstly, there should be some central party that determines who gets selected, which creates a single point of failure which is exactly what Bitcoin tries to avoid. Secondly, one can `fake' having many computers or cheaply create many IP-addresses, whereas it is impossible to solve the cryptographic puzzles without doing the work. The above explanation should clarify the idea behind, and the name of, the proof-of-work protocol.

In addition to the mining reward, the miner receives the transaction fees corresponding to all transactions in the block. When users want to make a Bitcoin transaction, they send the details of the transaction, along with a fee, to a node in the Bitcoin network. This node will spread the transaction to other nodes. When a transaction is known within the network, we say that it is inside the \textit{mempool} (short for memory pool). Miners are free to put any transactions from the mempool they want into the new block, but they obviously have the incentive to put the transactions in the block with the highest \textit{fee density} (i.e., highest fee per unit of transaction data, often stated in \textit{Satoshi per byte}, where `Satoshi' refers to $10^{-8}$ Bitcoin). The size of a transaction depends on the number of incoming transactions (inputs) that are linked to outgoing addresses (outputs). This implies that the size of a transaction is not fixed, and hence the number of transactions that fit in a block varies.

The blocks that are added to the chain have a maximum size of 1 Megabyte (MB): this is hardcoded in the Bitcoin software and ensures that the network remains stable. The difficulty of the cryptographic puzzle is periodically adjusted in such a way that a block is mined on average every 10 minutes (assuming that the so called \textit{hash rate}, i.e. computational power, of the network remains fixed). This leads to a transaction capacity of about 7 transactions per second (depending on the size of the transactions). The time until a block is found is approximately exponentially distributed \cite{btcfee}, which can be reasoned as follows. Miners search a very large search space at random. Each trial is independent of the previous, and each trial has an equal probability of success. Therefore, the number of trials until success has a geometric distribution. Since the success probability is very small and each trial takes in principle equally long, the geometric distribution can be approximated very closely by its continuous counterpart: the exponential distribution. The miner who solves the puzzle first, will take a time that is the minimum of all exponential times of all miners, which is again exponentially distributed. However, with data ranging over long enough time-spans, the hypothesis that blocks arrive after an exponential time can be falsified \cite{Taylor}. This can be explained by the fact that the hash rate of the whole network typically changes significantly in between difficulty adjustments.

\vb

During times of congestion, a Bitcoin user that wants to make a transaction is interested in the relation between the fee and confirmation time characteristics. The most obvious characteristics are the expected or `worst-case' confirmation times (which means that the transaction will be confirmed with, say, $95\%$ probability at that time). These quantities can be easily derived from the confirmation time distribution, which is the focus in this paper.

\subsection{Current literature and algorithms}
In the \textit{Bitcoin Core} 0.15 release there is a function called \textit{estimatesmartfee}. The input to this function is the maximum number of blocks after which you want your transaction to be included, and the output is the estimated fee that achieves this with a 95\% probability. The estimation is purely data driven. The procedure keeps track of the number of transactions in a particular fee density bucket that made it into into the blockchain after a given number of blocks. Exponential smoothing is applied (with smoothing factor $0.998$), to give higher weight to more recent transactions. For more details, see \cite{btcintro}. This algorithm will provide a reasonable benchmark, but one may expect improvements if the \textit{current} mempool is also taken into account. Moreover, a purely data-driven approach only extrapolates from past data and is therefore lagging behind if the environment changes quickly. A model-based approach will be more flexible in adjusting to current circumstances.

\vb

Alternatively, there is an estimator based on simulation, see \cite{btcfeeest1} and \cite{btcfeeest2}. This algorithm is model based, and assumes arrivals and batch services at Poisson epochs. The transaction size (in bytes) and transaction fee are sampled from a joint distribution based on historic blocks. This algorithm has been backtested on historical data, and appears to performs well. A disadvantage is that simulations tend to be time consuming. In our approach we try to apply mathematical analysis in order to gain more insight and to avoid computationally burdensome simulations.

\vb

The paper \cite{btcfee} provides an analytical probabilistic modeling approach for confirmation times, where the transaction process is seen as a $M/M^N/1$ queue with priorities and preemption. The stationary time until confirmation is analyzed. This provides understanding of the inner working of the Bitcoin blockchain on a system level, but it does not uncover the dynamics from a user perspective. For example, due to stationarity assumptions, the paper does not take into account information such as the current mempool and it is ignored that transactions have different sizes. Our paper aims to include these features, as they play an important role for predicting confirmation times.

\subsection{Our approach and contributions}
There are two sources of uncertainty of when a transaction will be confirmed: future block arrivals and future transaction arrivals. Since the average number of transactions per block is of the order $\sim 2000$, there will also be of the order $2000$ arrivals per batch departure in heavy-traffic situations. On such a scale, the main source of randomness is the block arrival process. Therefore, it is natural to model the amount of data with a higher fee as a Cram\'er-Lundberg process; this is a process well known from insurance mathematics and introduced in Section \ref{sec:btcmodel} below. In such models, the process slopes upwards with a fixed rate and there are jumps down at Poisson epochs. Typically, in the literature, it is assumed that the jumps down are phase type (e.g.\ exponential or Erlang), in which case relatively explicit results can be obtained with regards to the time until confirmation. However, in the case of Bitcoin, the jumps will be of size at most 1 (MB), since that is the max capacity of a block. We assume that all blocks will be completely filled whenever possible, which coincides with the miner's best interests. Cram\'er-Lundberg models with deterministic jumps down did not receive much attention in the literature, as results tend to be much more involved. One can appeal to general results involving spectrally negative L\'evy processes, of which our model is a special case, or expressions for the busy period in an $M/D/1$ queue starting in a level $x$ (cf.\ \cite{Kyprianou} or \cite{Pakes}). In either case, those results involve a double Laplace transform with respect to initial workload level and time, of which the inversion is computationally involved and potentially numerically unstable. Therefore we aim to derive more explicit and simple results. Firstly, we derive a lower bound on the probability that the transaction will be confirmed within a specified number of blocks. This approximation turns out to be close to simulated values when the traffic is relatively low. However, in heavy-traffic situations, the lower bound is very loose. Therefore, we also develop a second approach, based on a `corrected' diffusion approximation. 

In the diffusion approximation approach, we replace the Cram\'er-Lundberg process by a Brownian motion with the same drift, variance and starting point. This considerably simplifies the analysis for the confirmation time, as it is known that the hitting time of a Brownian motion with a drift towards the boundary has an inverse Gaussian distribution. Subsequently, we note that the approximation can be improved substantially by doing a `continuity correction'. This correction entails that the Brownian motion does not start in the same starting point as the Cram\'er-Lundberg process, but it is adjusted by the expected undershoot (i.e., the expectation of the value that process hits whenever it crosses zero from above for the first time). A computational procedure to calculate the expected undershoot is provided. Subsequently, the improvement as a result of the continuity correction is assessed by simulations.

\subsection{Outline}
In Section \ref{sec:btcmodel} we introduce our notation and the Cram\'er-Lundberg model, and we state the mathematical objective. In Section \ref{sec:btcresults} we describe the main theory: in Section \ref{sec:btclowerbound} we compute a lower bound to the probability of confirmation within a certain number of blocks. This bound is relatively tight as traffic is low, but its performance is poor in heavy traffic. Subsequently, in Section \ref{sec:btcdiffusion approximation} we consider a diffusion approximation which is particularly suitable for heavy-traffic regimes, but also works remarkably well in moderate-traffic situations. It is described how the approximation is corrected using the expected undershoot. Finally, in Section \ref{sec:btcsimulation}, the relative performance of the (corrected) Brownian approximations and the (simulated) Cram\'er-Lundberg model is considered. 

\section{Model}
\label{sec:btcmodel}
Given a transaction with a particular fee, let $X(t)$ denote the position of that transaction in the mempool at time $t$, if it is ordered by fee density. We suppose that $X(t)$ satisfies a Cram\'er-Lundberg type of model, i.e.:
\[
X(t) = x_0 + ct - P(t),
\]
where the process is scaled such that blocks have size $1$ and the time is scaled such that the Poisson process $P(t)$ has interarrival times of mean $1$, $x_0$ is the initial position in the mempool upon entering, and $c\geq0$ is the rate at which transaction data with higher priority arrives. Once the process hits zero, it means that the transaction under consideration has been confirmed. Thus $\tau=\inf\{t\geq0: X(t)=0\}$ means that the transaction has been confirmed after $\tau$ time units. The main mathematical objective in this paper is to recover the distribution of $\tau$. To ensure that $\Pb(\tau<\infty)=1$, we assume that $\E X(1)<0$. This is not restrictive, as users want to have their transactions confirmed with probability one at some point in the future, which is achieved by paying a sufficiently high fee.  

\begin{figure}[ht!]
\centering
\includegraphics[width=0.9\linewidth]{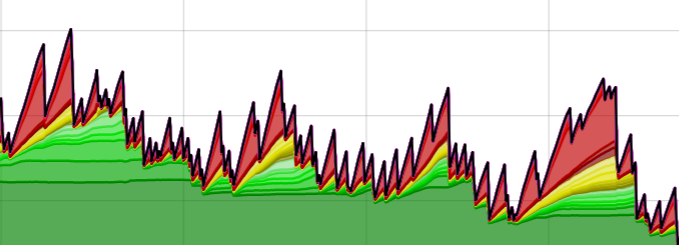}

\vspace{0.5cm}

\includegraphics[width=\linewidth]{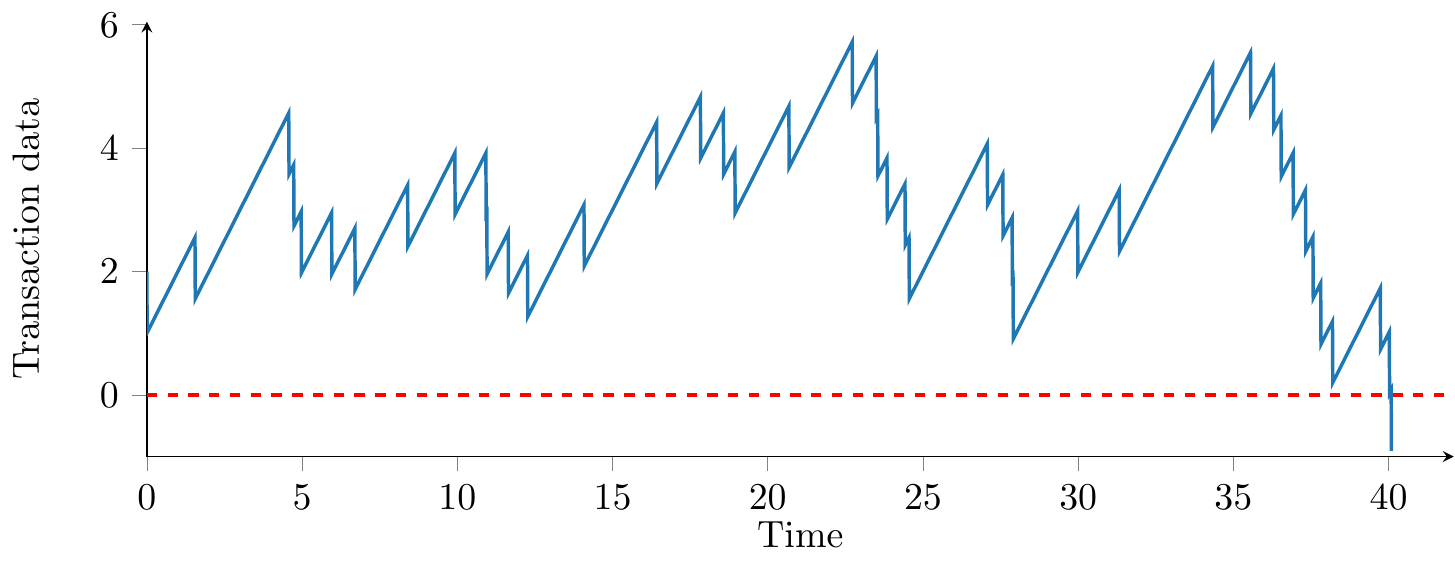}
\caption{Top: actual mempool data over time, as retrieved from \cite{hoenicke}, during a 16 hour timeframe in January 2018. Bottom: a sample path of the Cram\'er-Lundberg model.}
\label{fig:btcmempoolcomparison}
\end{figure}

A sample path of the $X(\cdot)$ process is shown in Figure \ref{fig:btcmempoolcomparison}, which is compared to a snapshot of true mempool data. Notice the similarity between how these two processes behave. It appears that the slope in the true mempool process is approximately constant in this time window, corresponding to our modeling assumption.

\section{Mathematical results}
\label{sec:btcresults}
In this section, we first derive a lower bound, which will be tight in light-traffic situations (as will be demonstrated in Section \ref{sec:btcsimulation}). Subsequently, we describe the diffusion approximation approach, which involves deriving an algorithm that can determine the expected overshoot of the Cram\'er-Lundberg model with deterministic jumps.

\subsection{A lower bound}
\label{sec:btclowerbound}
Let $T_1, T_2,\ldots$ be a set of i.i.d., exponentially distributed, random variables with unit mean, that represent the time until the next block is found. Define $X_i = X(T_i)$. Then it holds that
\[
X_{n+1} = X_n+c T_n-1.
\]
For now, we assume that there is no sticky boundary at zero: the process can become negative and positive again. This means that, in the model, the transaction could switch from confirmed to unconfirmed. Therefore the actual probability of confirmation will be larger than the probability that the $X(\cdot)$ process is negative, i.e. we find a lower bound. Consider the following proposition.

\begin{proposition}
\label{prop:btclowerbound}
Suppose that given a fee density $\phi$, one can determine the rate $c=c(\phi)$ at which transactions with a fee density higher than $\phi$ arrive and one can observe the initial place in the mempool $x_0=x_0(\phi)$, which is ordered by fee density. Then the probability of confirmation in at most $n$ blocks, is bounded by
\[
\Pb(\text{confirmation in $\leq n$ blocks})\geq 1-\sum_{k=0}^{n-1} \f{y_n^{k}}{k!}e^{-y_n},\quad\text{with}\quad y_n:=\max\{\f{n-x_0}{c},0\}.
\]
\end{proposition}
\begin{proof}
The sum of $n$ independent exponentially distributed random variables is Erlang distributed. In particular, with $y_n(x):=\frac{n-x_0+x}{c}$, it holds that
\begin{align*}
\Pb(X_n>x) &= \Pb(x_0-n+c \sum_{i=1}^n T_i >x) = \Pb\left(\sum_{i=1}^n T_i > \f{x+n-x_0}{c}\right)\\
&=
\begin{cases}
1&\text{if } \frac{x+n-x_0}{c}\leq 0\\
\sum_{k=0}^{n-1} \f{y_n(x)^{k}}{k!}e^{-y_n(x)} & \text{if } \frac{x+n-x_0}{c}>0.
\end{cases}
\end{align*}
Since $X_n\leq 0$ implies ``$\textrm{confirmation in }\leq n \textrm{ blocks}$'', it follows that
\begin{align*}
\Pb(\textrm{confirmation in }\leq n \textrm{ blocks}) &\geq \Pb(X_n\leq0) = 1- \Pb(X_n>0)\\
&=
\begin{cases}
0 & \text{if } n \leq x_0\\
1-\sum_{k=0}^{n-1} \f{y_n(0)^{k}}{k!}e^{-y_n(0)} & \text{if } n > x_0,
\end{cases}
\end{align*}
which yields the result.
\end{proof}


This gives us a lower bound which can be calculated very efficiently. The simulations in Section \ref{sec:btcsimulation} reveal that the bound is quite tight when traffic is light (say, $c\leq 0.5$), but the bound becomes poor in heavy-traffic situations.

\subsection{Diffusion approximation}
\label{sec:btcdiffusion approximation}
To simplify the analysis, we could replace the Cram\'er-Lundberg process by a Brownian motion with a matching drift and variance. Due to our time and space scaling, it is easy to see that the variance of the Cram\'er-Lundberg process equals unity. It is well known that the hitting time of a Brownian motion with a drift towards a constant boundary has an inverse Gaussian distribution. In certain cases, however, this approximation will perform rather poorly. For instance, if the starting position $x_0=0$ and the slope $c=0$, then the hitting time is simply the realization of the first exponential interarrival time of the Poisson process. However, due to its properties, the Brownian motion will immediately go below the $0$ boundary with probability one. We can `correct' for this, by noting that the undershoot of the Cram\'er-Lundberg process below level $0$ will be of size $1$ in this case, and by letting the Brownian motion start in level $1$. For general $c>0$, the expected undershoot can still be calculated via the theorem and algorithm stated below. The idea of `correcting' a diffusion approximation of a random walk by changing the initial point (and the drift) of the approximating Brownian motion is not new, cf.\ e.g.\ \cite{siegmund}. We contribute to this line of research by calculating the expected overshoot explicitly for this particular Cram\'er-Lundberg model. In the following theorem and its proof, we write $n_x$ to denote $\lceil x\rceil$, i.e.\ the smallest integer greater than or equal to $x$.

\begin{theorem}
\label{thm:undershoot}
The undershoot $S_x$ below level $0$ of the Cram\'er-Lundberg process that starts in level $x$, with $P(t)$ a Poisson process with rate $1$, satisfies, for $0\leq x\leq 1$,
\begin{equation}
\label{eq:Sx01}
\E S_x= 1-c+e^{-\f{1-x}{c}}(c+\E S_1) - x,
\end{equation}
and for $x>1$, 
\begin{equation}
\label{eq:Sx1}
\E S_x = \int_0^{\f{n_x-x}c} \E S_{x+cy-1} e^{-y}\dif y + e^{-\f{n_x-x}c}\E S_{n_x}.
\end{equation}
\end{theorem}
The proof of this theorem can be found after the algorithm below. Note that Algorithm 1 is ideally implemented in a programming language that supports symbolic mathematics, such as \textsc{Mathematica}. 

\begin{algorithm}
\label{alg:sx}
\textsc{Initialize}: Set $c$ and $n_{\textrm{max}}$ (comment: $n_{\textrm{max}}$ needs to be large enough so that $\E S_{n_\textrm{max}}\approx \lim_{x\to\infty} \E S_x$ and the $x$ in value of interest $\E S_x$ satisfies  $x<n_{\textrm{max}}$)\\
\textsc{Step} 1: 
\begin{itemize}
\item Set $n=1$. While $n\leq n_{\textrm{max}}$: express $f_n(x) = \int_{0}^{\f{n-x}c} e^{-y} f_{n-1}(x+yc-1)\dif y + e^{-\f{n-x}c} a_n$  (comment: $a_n$ is a constant that is yet to be determined)
\end{itemize}
\textsc{Step} 2: Determine $a_n$
\begin{itemize}
\item Set $n=1$. While $n\leq n_{\textrm{max}}-1$: Solve $a_n = f_{n+1}(n)$ for $a_{n+1}$ and increment $n\leftarrow n+1$ (comment: this step expresses $a_n$ into $a_{n-1}$, into $a_{n-2}$, $\ldots$, and finally into $a_0$)
\item Set $a_{n_{\textrm{max}}-1}=a_{n_{\textrm{max}}}$ and solve for $a_0$ (comment: since each $a_n$ was expressed into $a_0$, we have thus determined each $a_n$)
\end{itemize}
\textsc{Output}: $\E S_{x}$ is given by $f_{n_x}(x)$ 
\end{algorithm}


\begin{proof}[Proof of Theorem \ref{thm:undershoot} and Algorithm \ref{alg:sx}]
By the memorylessness property of the exponential distribution, we have for $0\leq x \leq 1$
\begin{align*}
\E S_x &= \E[S_x | c T_1<1-x]\Pb(c T_1<1-x) + \E[S_1|c T_1\geq1-x] \Pb(cT_1\geq1-x)\\
&=\E[S_x 1_{\{c T_1<1-x\}}]+ e^{-(1-x)/c}\E[S_1]\\
&=\E[(1-c T_1-x) 1_{\{c T_1<1-x\}}]+ e^{-(1-x)/c}\E[S_1]\\
&=(1-x)\Pb(c T_1<1-x)-\E[c T_1 1_{\{c T_1<1-x\}}]+ e^{-(1-x)}\E[S_1]\\
&=(1-x)\Pb(c T_1<1-x)-c \int_0^{\f{1-x}{c}} y e^{-y} \dif y+ e^{-\f{1-x}{c}}\E[S_1]\\
&=(1-x)(1-e^{-\f{1-x}{c}}) - c - (x-1-c)e^{-\f{1-x}{c}}+ e^{-\f{1-x}{c}}\E[S_1]\\
&=1-c+e^{-\f{1-x}{c}}(c+\E S_1) - x,
\end{align*}
so in particular
\[
\E S_0 = 1-c+e^{-1/c}(c+\E S_1).
\]
With $n_x=\lceil x\rceil$ and using a similar argument as above, we have for $n_x-1 < x \leq n_x$,
\begin{align*}
\E S_x &= \E[S_x | cT<n_x-x]\Pb(cT<n_x-x) + \E[S_{n_x}|cT\geq n_x-x] \Pb(cT\geq n_x-x)\\
&=\E[S_x | cT<n_x-x]\Pb(cT<n_x-x) + e^{-\f{n_x-x}c} \E S_{n_x}\\
&=\E[S_x 1_{\{cT<n_x-x\}}] + e^{-\f{n_x-x}c} \E S_{n_x}\\
&=\int_{0}^{\f{n_x-x}c}\E S_{x+c y-1} e^{-y}\dif y + e^{-\f{n_x-x}c} \E S_{n_x}.
\end{align*}
This finishes the proof of the proposition. The algorithm is a straightforward application of the proof. 
\end{proof}

\begin{figure}[ht!]
\centering
\includegraphics[width=0.49\textwidth]{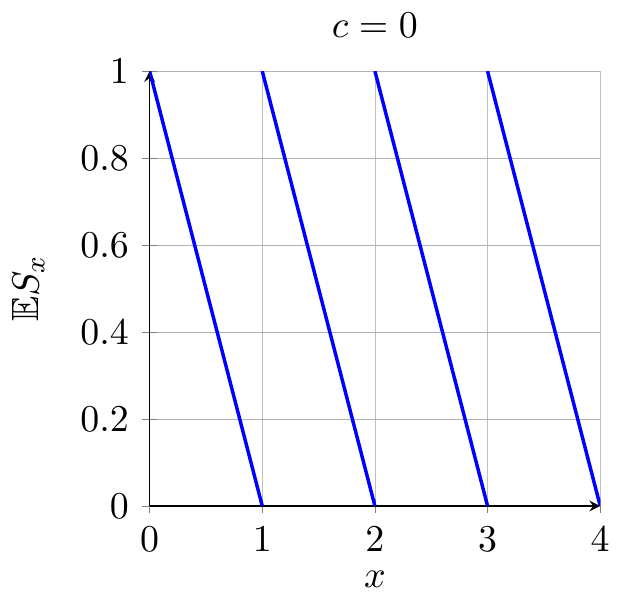}
\includegraphics[width=0.49\textwidth]{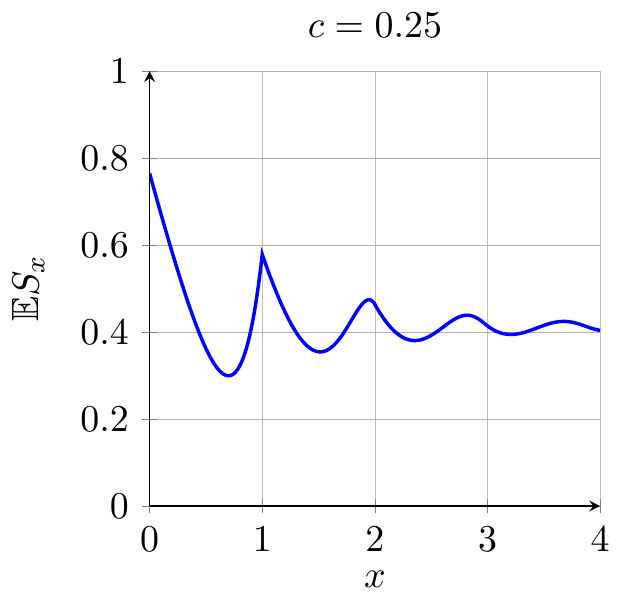}

\vspace{0.25cm}

\includegraphics[width=0.49\textwidth]{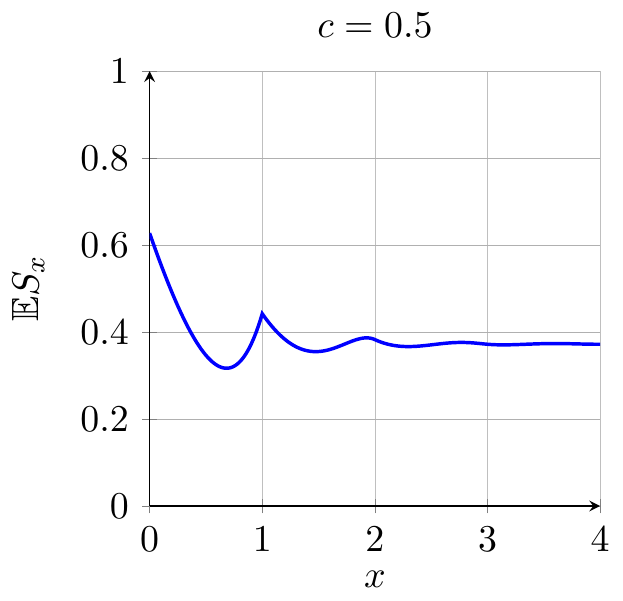}
\includegraphics[width=0.49\textwidth]{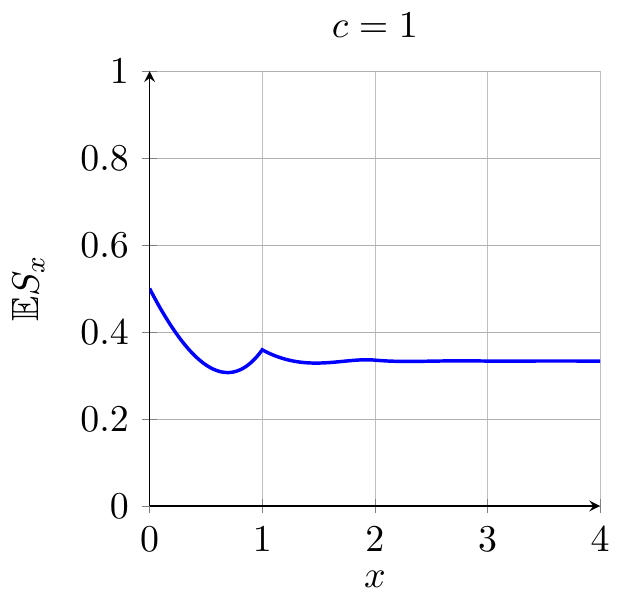}
\caption{The expected undershoot $\E S_x$ as a function of the initial position $x$. Note that as traffic becomes more intense, the undershoot is less variable as function of the starting position, and converges faster to its limiting value.}
\label{fig:btcexpectedundershoot}
\end{figure}

Using a slightly different approach than Algorithm \ref{alg:sx}, it is even possible to derive exact error bounds on the expected undershoot after $n$ iterations. This can be done by using the expression of $\E S_{n-1}$ in terms of $\E S_n$, for all $n$. 

\begin{example}
Using the theory above in the special case $c=1$, it can be recursively calculated that
\begin{align*}
\E S_0 &= \f1e + \f1e \E S_1\approx 0.368 + \f1e \E S_1 \\
\E S_0 &= \f2{e(e-1)} + \f1{e(e-1)} \E S_1\approx 0.428 + \f1{e(e-1)} \E S_2\\
\E S_0 &= \frac{4}{(e-1) (1+2 (e-2) e)} + \frac{2}{e-4 e^2+2 e^3} \E S_3\\
&\approx 0.4746 + \frac{2}{e-4 e^2+2 e^3} \E S_3\\
\E S_0 &= \frac{e (15+e)-24}{(e-1) e (1+e (5+2 (e-3) e))} + \frac{(e-1)^2 \E S_4}{e^3 (1+e (5+2 (e-3) e))}\\
&\approx 0.496 + 0.014 \E S_4 .
\end{align*}
Since clearly $0\leq \E S_n \leq 1$, for every $n$, we know that the error is at most the prefactor of $\E S_{n+1}$. Therefore, after four iterations, it follows that
\[
\E S_0 \in [0.496,0.510].
\] 
This corresponds to the value $0.500\pm 0.001$ (where the value after $\pm$ indicates the standard error) that we found by simulation. 
\end{example}

An implementation of Algorithm \ref{alg:sx} in \textsc{Mathematica} leads to Figure \ref{fig:btcexpectedundershoot}, where $c$ is varied ranging from $c=0$ to $c=1$. In the particular case $c=1$ (which implies zero drift and is known as `no safety loading' in the insurance literature), the limiting density of the undershoot can be explicitly calculated, as the initial position tends to infinity. Indeed, with $Y$ the random variable corresponding to jump sizes, \cite[Theorem 6]{undershoot} states that the limiting density of the undershoot in case of zero drift and finite variance jump sizes satisfies
\[
f(y) = \frac{\int_0^\infty z \Pb(Y>z+y)\dif z}{\int_0^\infty z \Pb(Y>z)\dif z}.
\]
In \cite{bscmeeles} it is derived under the same assumptions that
\[
\Pb\left( \lim_{x\to\infty} S_x>y\right) = \frac{\E[\max\{0,(Y-y)^2\}]}{\E Y^2},
\]
which in case deterministic jumps of unit size simplifies to
\[
\Pb\left( \lim_{x\to\infty} S_x>y\right) = (1-x)^2,\quad\text{for } 0\leq x\leq 1.
\]
It follows that 
\[
\lim_{x\to\infty} \E[S_x] = \f13,
\]
which is indeed the value to which $\E S_x$ appears to converge, for $x$ large, in the plot of Figure \ref{fig:btcexpectedundershoot} corresponding to the case of $c=1$. 

\section{Simulation}
\label{sec:btcsimulation}
\begin{figure}[ht!]
\centering
\includegraphics[width=0.49\textwidth]{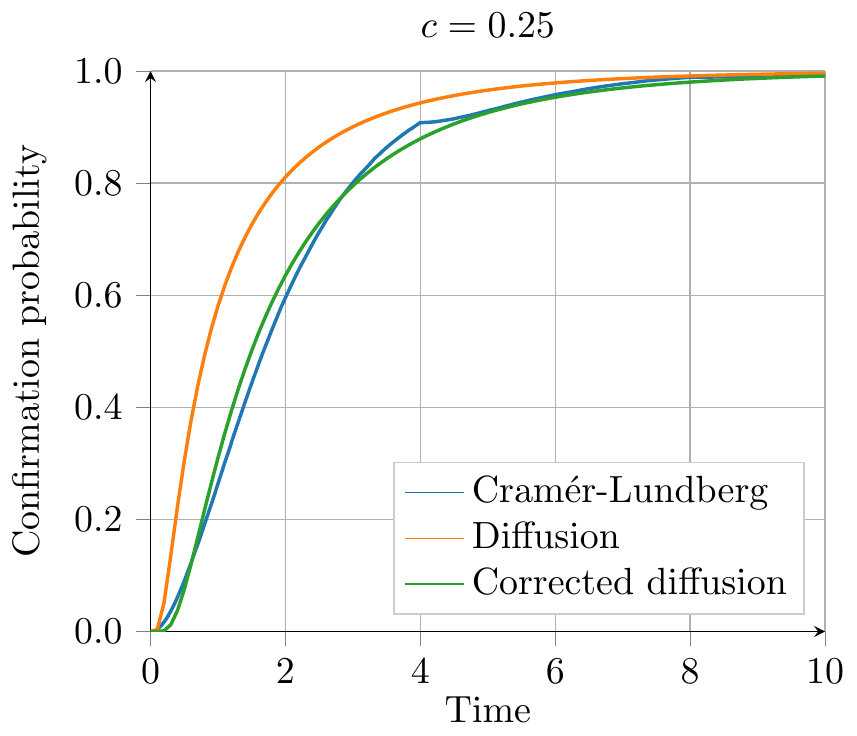}
\includegraphics[width=0.49\textwidth]{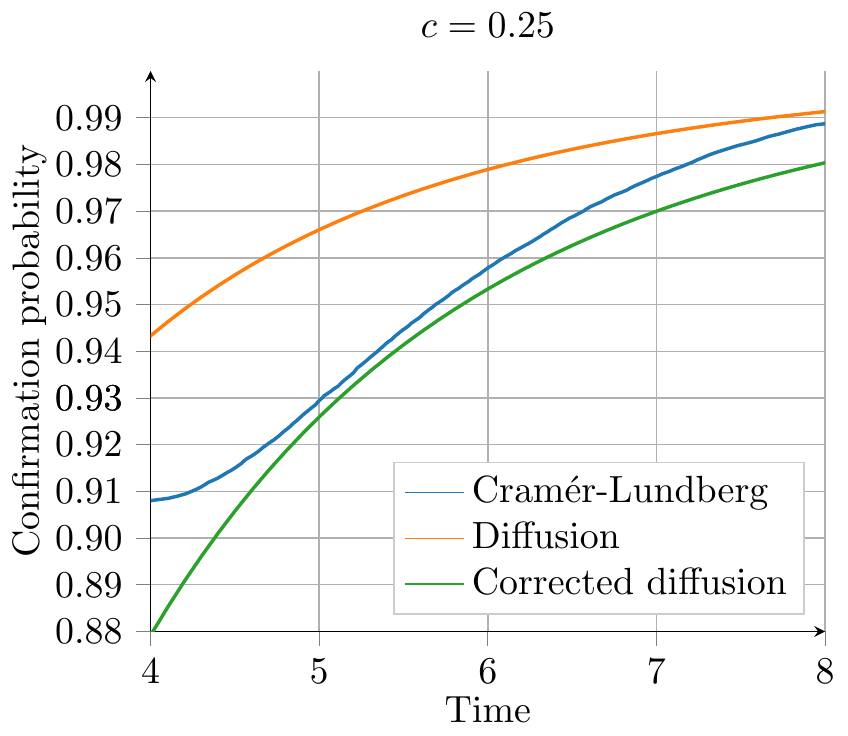}
\includegraphics[width=0.49\textwidth]{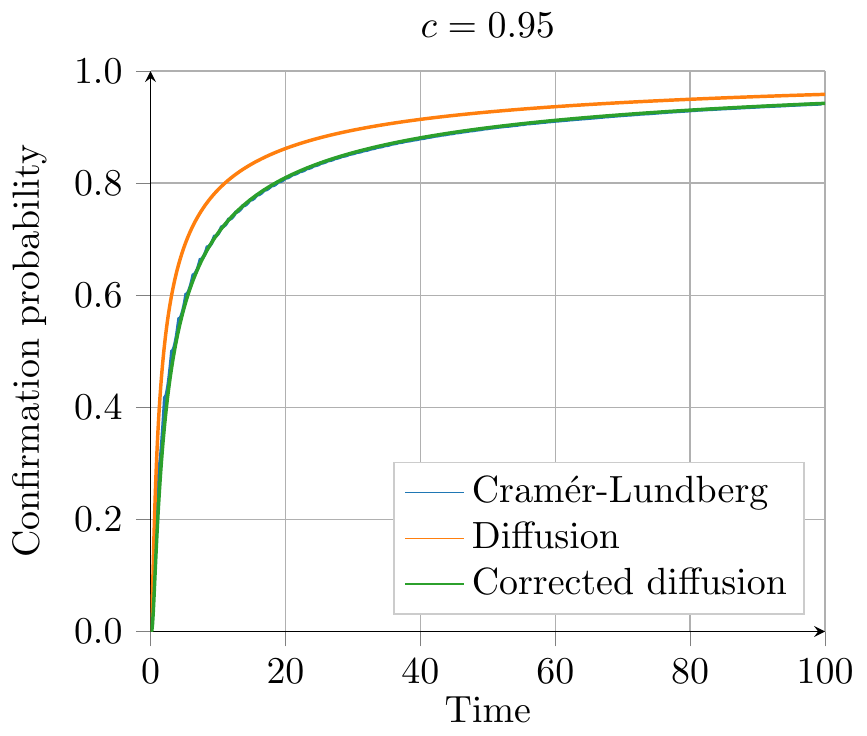}
\includegraphics[width=0.49\textwidth]{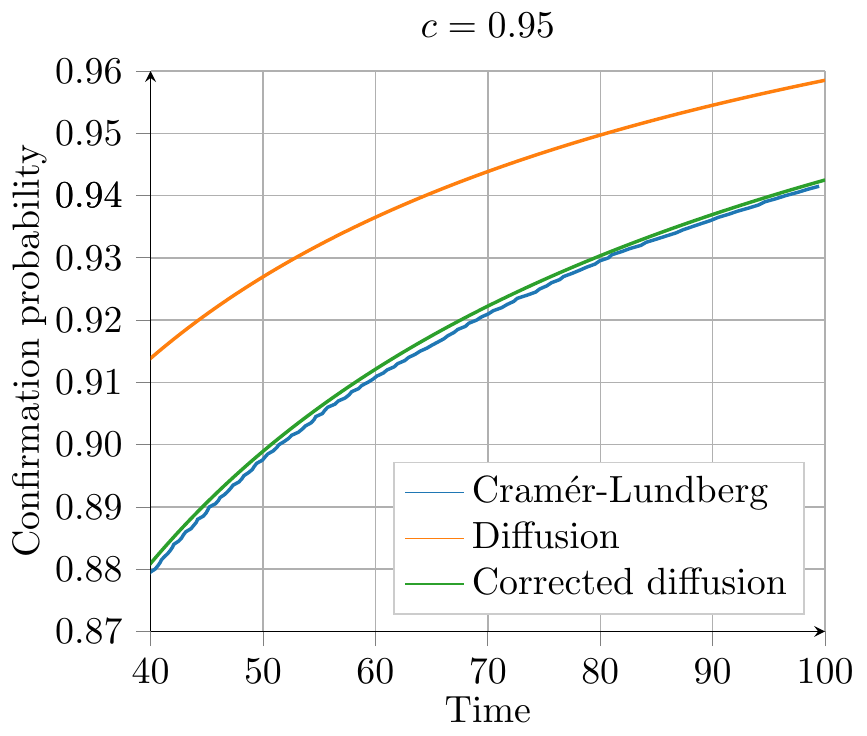}
\caption{In each graph we compare the simulated confirmation times of the Cram\'er-Lundberg process, with starting point $x=1$, to the confirmation times corresponding to the diffusion and corrected diffusion approximations. The two right-hand plots are zoomed in versions of the left-hand plots, the top two plots are for light traffic $c=0.25$ and the bottom two plots are for heavy traffic $c=0.95$.}
\label{fig:btcbmvscor}
\end{figure}

\begin{table}
\centering
\begin{tabular}{cccc}
$n$ & Sim mean & Sim st.\ dev. & Lower bound\\
\hline
4&0.000&0.000&0.000\\
5&0.370&0.005&0.371\\
6&0.811&0.004&0.809\\
7&0.956&0.002&0.954\\
8&0.991&0.001&0.990
\end{tabular}
\caption{Probability that the number of blocks until confirmation is at most $n$, where $c=0.25$ and $x=4$, based on 300,000 simulations.}
\label{tab:btclowerboundc025}

\vb
\centering
\begin{tabular}{cccc}
$n$ & Sim mean & Sim st.\ dev. & Lower bound\\
\hline
1 & 0.000 & 0.000& 0.000\\
2 & 0.593 & 0.006& 0.594\\
3 & 0.797 & 0.004& 0.762\\
4 & 0.890 & 0.003& 0.849\\
5 & 0.937 & 0.002& 0.900\\
6 & 0.962 & 0.002& 0.933\\
7 & 0.977 & 0.002& 0.954\\
8 & 0.985 & 0.002& 0.968
\end{tabular}
\caption{Probability that the number of blocks until confirmation is at most $n$, where $c=0.5$ and $x=1$, based on 300,000 simulations.}
\label{tab:btclowerboundc050}

\vb
\centering
\begin{tabular}{cccc}
$n$ & Sim mean & Sim st.\ dev. & Lower bound\\
\hline
1 & 0.000 & 0.000& 0.000\\
2 & 0.385 & 0.004& 0.385\\
3 & 0.560 & 0.005& 0.498\\
4 & 0.662 & 0.004& 0.616\\
$\vdots$&$\vdots$&$\vdots$&$\vdots$\\
18 & 0.954 & 0.002& 0.853\\
19 & 0.958 & 0.002& 0.863\\
20 & 0.962 & 0.002& 0.872
\end{tabular}
\caption{Probability that the number of blocks until confirmation is at most $n$, where $c=0.75$ and $x=1$, based on 300,000 simulations.}
\label{tab:btclowerboundc075}
\end{table}

In this section we compare the diffusion and corrected diffusion approximation to the simulated Cram\'er-Lundberg process, both in a heavy-traffic ($c=0.95$) and a light-traffic ($c=0.25$) regime. In most cases we suppose that the starting position is $x=1$: depending on your preferences and the traffic, it is reasonable in practice to pick your fee such that your initial position in the queue is of the order of one block. By Algorithm \ref{alg:sx} we determined that the correction in the Brownian motion is approximately $\E S_1 = 0.36403$ when $c=0.95$ and $\E S_1 =0.57833$ when $c=0.25$. Consider Figure \ref{fig:btcbmvscor} for the results. It is quickly observed that the corrected diffusion approximation is generally significantly more accurate than the uncorrected diffusion approximation. Moreover, it appears that in the heavy-traffic situation the diffusion approximation is more accurate than in the light-traffic situation. 

In the light-traffic situation we can resort to the lower bound algorithm, of which the bounds will be quite tight in light traffic. Due to the nature of Proposition \ref{prop:btclowerbound}, we have to compare the confirmation probability after $n$ blocks (rather than confirmation probability after a time $t$). We report the results corresponding to the lower bound in Tables \ref{tab:btclowerboundc025}, \ref{tab:btclowerboundc050} and \ref{tab:btclowerboundc075}. We considered three cases: $c=0.25$, $c=0.50$ and $c=0.75$. In Tables \ref{tab:btclowerboundc050} and \ref{tab:btclowerboundc075} we chose a starting point $x=1$, but in Table \ref{tab:btclowerboundc025} we increased the starting point to $x=4$ (with $c=0.25$ and $x=1$ most confirmations would simply be after two blocks). Note that the lower bound is indeed below the simulated values, or above it within a standard deviation. In addition, it can be seen that the lower bound becomes increasingly loose as $c$ increases. With practical applications in mind, we could say that the lower bounds are tight enough if and only if $c\leq 0.5$, especially when one is interested in confirmation probabilities within the $90\%$ to $100\%$ range. 

\newpage
\section{Further research and concluding remarks}
There are several ways to continue this research, most notably in the following two ways:
\begin{itemize}
\item In this research we assumed the existence of a single `true' mempool, but in fact each node keeps track of its own mempool. The mempools are continuously updated with each other, but in reality there is some lag before new transactions spread through the entire network. Therefore it can happen that a miner finds a block while being unaware of some of the most recent transactions, implying that they will not be included in the block even though their fees may be sufficient. This effect is hard to model and arguably quite small, and therefore we ignored it. The model should be backtested on data to find out about the practical applicability. Fortunately, transaction data and mempool data is publicly available: it can be collected by setting up a Bitcoin node.
\vspace{-0.3cm}
\item It would be interesting to solidify the mathematical fundementals of this paper by proving the following conjecture: \textit{As the traffic ($c\uparrow1$) and time are scaled appropriately, the Cram\'er-Lundberg process with deterministic jumps of fixed size, converges to a Brownian motion. As a result, the hitting time of the Cram\'er-Lundberg process converges to an inverse Gaussian distribution.}
\end{itemize}

\section*{Acknowledgements}
The author thanks Onno Boxma and Michel Mandjes for their proofreading and helpful suggestions. The research for this paper is funded by the NWO Gravitation Project NETWORKS, Grant Number 024.002.003. 

{\small
\bibliographystyle{plain}
\bibliography{biblio}}
\end{document}